\documentclass[a4paper,UKenglish,cleveref, autoref]{lipics-v2019}

\title{Reverse derivative categories 
\footnote{This paper is the result of a joint working session that the authors 
participated in at the Foundational Methods in Computer Science workshop
in June, 2019.}}

\author{Robin Cockett}{University of Calgary, Department of Computer Science, Canada}{robin@ucalgary.ca}{}{Partially supported by NSERC (Canada)}
\author{
  Geoffrey Cruttwell
}{Mount Allison University, Department of Mathematics and Computer Science, Canada}{gcruttwell@mta.ca}{}{Partially supported by NSERC (Canada)}
\author{Jonathan Gallagher}{Dalhousie University, Department of Mathematics and Statistics, Canada}{jonathan.gallagher@dal.ca}{}{Supported in part by NSERC and AARMS (Canada)} 
\author{Jean-Simon Pacaud Lemay}{University of Oxford, Department of Computer Science, UK}{jean-simon.lemay@kellogg.ox.ac.uk}{}{Supported by Kellogg College, the Clarendon Fund, and the Oxford Google-DeepMind Graduate Scholarship (UK)}
\author{Benjamin MacAdam}{University of Calgary, Department of Computer Science, Canada}{benjamin.macadam@ucalgary.ca}{}{Partially supported by NSERC (Canada)}
\author{Gordon Plotkin}{Google Research}{gdp@inf.ed.ac.uk}{}{Supported by ESPRC (UK)}
\author{Dorette Pronk}{Dalhousie University, Department of Mathematics and Statistics, Canada}{dorette.pronk@dal.ca}{}{Partially supported by NSERC (Canada)}

\authorrunning{CCGLMPP}
\Copyright{
  Robin Cockett,
  Geoffrey Cruttwell,
  Jonathan Gallagher,
  Jean-Simon Lemay,
  Benjamin MacAdam,
  Gordon Plotkin,
  and 
  Dorette Pronk
}
\acknowledgements{We thank Robert Seely for participating in the discussion on 
reverse differentiation with us.}

\usepackage{etex}

\usepackage{mathrsfs}

\usepackage[all,2cell,cmtip]{xy}
\UseTwocells

\usepackage[dvipsnames]{xcolor}
\usepackage{enumerate}
\usepackage{BarrTo}
\usepackage{amsfonts,amsmath,amssymb,amsthm}
\usepackage{proof}
\usepackage{tikz-cd}

\usepackage{tikz}
\usetikzlibrary{decorations.pathmorphing}
\tikzset{snake it/.style={decorate, decoration=snake}}
\usetikzlibrary{arrows,calc}  

\usepackage[all]{xy}

\usepackage{microtype}

\usepackage{stmaryrd}

\newcommand{\R}{\mathsf{R}}
\newcommand{\D}{\mathsf{D}}

\newcommand{\x}{\times}

\newcommand{\ox}{\otimes}
\newcommand{\blank}{\underline{~}}

\newcommand{\X}{\mathbb{X}}

\newcommand{\<}{\left\langle}
\renewcommand{\>}{\right\rangle}

\newcommand{\Y}{\mathbb{Y}}

\newcommand{\RR}{\R^{(2)}}

\newcommand{\smooth}{\mathsf{Smooth}}

\renewcommand{\lim}{\mathsf{lim}}

\newcommand{\ex}{\mathsf{ex}}


\nolinenumbers

\begin{document}
\maketitle 

\begin{abstract}
The reverse derivative is a fundamental operation in machine learning and 
automatic differentiation \cite{tensorflow2015-whitepaper, griewank2012invented}.  This paper gives a direct axiomatization of a 
category with a reverse derivative operation, in a similar style to that given by \cite{journal:BCS:CDC}
for a forward derivative.  Intriguingly, a category with a reverse derivative 
also has a forward derivative, but the converse is not true.  In fact, we show explicitly what a forward derivative is missing: 
a reverse derivative is equivalent to a forward derivative with a
dagger structure on its subcategory of linear maps. Furthermore, we show that these linear maps form an additively enriched category with dagger biproducts.
\end{abstract}

\section{Introduction}

 

The use of derivatives and differentiation in programming and machine learning is becoming ubiquitous. 
As a result, there has been an increased interest in axiomatic setups for differentiation; in particular, categorical models for differentiation have become more central.  There are two types of derivative operations used in programming: the \emph{forward} derivative and the \emph{reverse} derivative.  From the programmer's perspective, it is much more common for the reverse derivative to play the central role due to its increased efficiency and improved accuracy when computing with functions from $\mathbb{R}^n$ to $\mathbb{R}$ (due to the so called cheap gradient principle). The importance of this principle was already recognized by Linnainmaa in 1976 \cite{linnainmaa} and was specifically used for back-propagation in multi-layer networks and deep learning.  This was further spelled out in detail in \cite{backpropagation}. Also, Tensorflow,  Google's new interface for expressing machine learning algorithms, uses the reverse mode of automatic differentiation as the basic building block minimizing cost functions \cite{tensorflow2015-whitepaper}. 

The categorical approaches to differentiation to date have all exclusively focused on the abstract properties of the forward derivative \cite{journal:BCS:CDC}.  This thus leaves a significant gap which needs to be filled: an axiomatic categorical setting for reverse differentiation.  The main goal of this paper is to introduce such a structure and explore some of its properties and consequences.

A ``Cartesian reverse differential category'' (a category equipped with a reverse derivative operation as introduced in this paper) is already a Cartesian differential category (the standard axiomatics for a category with a forward derivative).  We show that a category equipped with a reverse derivative also has a forward derivative (i.e., it has a Cartesian differential structure).  Moreover, a reverse differential category has a fibered dagger structure on its subcategory of linear maps, a structure which does not automatically exist in a Cartesian differential category.  Suitably axiomatized, we show that having such a dagger structure is enough to ensure that a Cartesian differential category structure gives a reverse differential category.  These results  provide a starting point to build categorical semantics of differential programming languages \cite{talk:plotkin-2018}, as they provide axiomatically enough structure to handle both forward and reverse derivatives.

The paper is structured as follows.  In section \ref{section:forward-der}, we recall the basic notation and definitions of a Cartesian differential category (``a category equipped with a forward derivative'').  We do this first to acclimatize the reader to the general style of this categorical definition, and to recall the structure of Cartesian left additive categories, which are necessary to define both forward and reverse differential categories.   In section \ref{section:reverse-der}, we introduce our definition of a reverse differential category.  We explore some of the important consequences of the definition noted above: (a) Cartesian differential structure, (b) how to define and work with linear maps in this setting, and (c) a dagger structure on the linear maps.  In section \ref{section:cont-dagger}, we show how to go back: given a Cartesian differential category with a ``contextual dagger'', we build a Cartesian reverse differential category.  There is much more work to be done with this structure and these ideas: in section \ref{section:conclusion}, we describe some of the ways in which this work can be extended, including allowing partial functions.  

As far as we are aware, this paper represents the first categorical axiomatization of the reverse derivative.  However, \cite{elliott-AD-ICFP} does have some related ideas.  There, the relationship between the reverse derivative and coproducts was noticed, and the author specified an internal category which satisfies some of the axioms of a Cartesian differential category in a functional programming language. This work expands that observation by developing the dagger biproduct structure using the reverse derivative and relating this to the dual of the simple slice fibration.



\section{Forward derivatives}\label{section:forward-der}

The standard setting for a ``category with a forward derivative'' is a Cartesian differential category, first introduced in \cite{journal:BCS:CDC}.  Following that paper, we write composition in diagrammatic order, so that $f$, followed by $g$, is written as $fg$.  

\subsection{Cartesian left additive categories}\label{section:clacs}

A Cartesian differential category first consists of a Cartesian left additive category, and so we begin by recalling this notion. Recall that a category $\X$ is said to be {\bf Cartesian} when there are chosen binary products $\times$, with projection maps $\pi_i$ and pairing operation $\langle - , - \rangle$, and a chosen terminal object $\mathbf{1}$, with unique maps $!$ to the terminal object. 

\begin{definition} A {\bf left additive category} \cite[Definition  1.1.1]{journal:BCS:CDC} is a category $\X$ such that each hom-set  is a commutative monoid, with addition operation + and zero maps 0, such that composition on the left preserves the additives structure in the sense that $x(f+g) = xf + xg$ and $x0=0$.  Maps $h$ which preserve the additive structure by composition on the right ($(x+y)h = xh + yh$ and $0h=0$) are called {\bf additive}. A {\bf Cartesian left additive category} \cite[Definition 1.2.1]{journal:BCS:CDC} is a left additive category $\X$ which is Cartesian and such that all projection maps $\pi_i$ are additive\footnote{Note that this a slight variation on the definition of a Cartesian left additive category found in \cite{journal:BCS:CDC}, but it is indeed equivalent.}. 
\end{definition}

Cartesian left additive categories can alternatively be defined as Cartesian categories in which each object $A$ canonically bears the structure of a commutative monoid with addition $+_A: A \x A \to A$ and zero $0_A: \mathbf{1} \to A$. 

\begin{example} Here are examples of Cartesian left additive categories that we will consider throughout this paper: 
\begin{enumerate}
\item  Any category with finite biproducts is a Cartesian left additive category where every map is additive.  And conversely, in a Cartesian left additive category where every map is additive, the finite product is a finite biproduct \cite[Proposition 1.2.2]{journal:BCS:CDC}. 
\item Let $R$ be a commutive rig (also known as a commutative semiring). Let $\mathsf{POLY}_R$ be the category of polynomials with coefficients in $R$; that is, the category whose objects are the natural numbers $n \in \mathbb{N}$ and where a map $n \to^{P} m$ is an $m$-tuple of polynomials $P := \langle p_1(\vec x), \hdots, p_m(\vec x) \rangle$, where $p_i(\vec x) \in R[x_1, \hdots, x_n]$ (the polynomial ring in $n$-variables over $R$). $\mathsf{POLY}_R$ is a Cartesian left additive category where composition is given by the standard composition of polynomials, the product on objects is given by the sum of natural numbers, and the additive structure is given by the sum of polynomials. 
\item Let $\mathbb{R}$ be the set of real numbers and let $\smooth$ be the category of smooth real functions, that is, the category whose objects are again the natural numbers $n \in \mathbb{N}$ and where a map $n \to^{F} m$ is a smooth function $\mathbb{R}^n \to^{F} \mathbb{R}^m$. $\smooth$ is a Cartesian left additive category where composition is given by the standard composition of smooth functions, the product on objects is given by the sum of natural numbers, and the additive structure is given by the sum of smooth functions. Note that a smooth map $\mathbb{R}^n \to^{F} \mathbb{R}^m$ is actually an $m$-tuple of smooth functions $F = \langle f_1, \hdots, f_m \rangle$, where $\mathbb{R}^n \to^{f_i} \mathbb{R}$ and therefore $\mathsf{POLY}_\mathbb{R}$ is a sub-Cartesian left additive category of $\smooth$.

\end{enumerate}

\end{example}

As not every map in a Cartesian left additive category is additive, the product $\times$ is not a coproduct, thus is \emph{not} a biproduct. However, it is still possible to define injection maps. So in a Cartesian left additive category,  define $\iota_0 := \<1,0\>:A \to A\x B$ and ${\iota_1 := \<0,1\>: B\to A\x B}$. For maps $A \to^{f} C$ and $B\to^{g} C$, we
define $\<f|g\> := \pi_0 f + \pi_1 g : A \x B \to C$, and finally for maps $A \to^{h} B$ and $C\to^{k} D$ we write $h\oplus k:= \<h\iota_0|k\iota_1\>: A \times B \to C \times D$. 
 Although this notation is suggestive, we again stress this is not part of a coproduct or biproduct structure.  However, in what follows we will define the category of linear maps where the above will witness a biproduct structure on that category. We leave the following lemma as an easy exercise to the reader: 

\begin{lemma}\label{lemma:clac-simple} In a Cartesian left additive category, $f\iota_0 + g \iota_1 = \<f,g\>$ and $h\oplus k = h\x k$. 
\end{lemma}

\subsection{Cartesian differential categories}\label{section:foward-der}
This section reviews Cartesian differential categories which provide the
semantics for forward differentiation \cite{journal:BCS:CDC}.

\begin{definition}
  A {\bf Cartesian differential category} \cite{journal:BCS:CDC} is a Cartesian left additive category 
  with a combinator $\D$, called the {\bf differential combinator}, which written as an inference rule is given by: 
$$ \infer{A\x A\to_{\D[f]} B}{A\to^{f} B}  $$ 
where $\D[f]$ is called the derivative of $f$, and such that the following equalities hold\footnote{Note that the order of variables is different here than in  \cite{journal:BCS:CDC}; here, we write the vector variable in the second component, as this more closely aligns with standard differential calculus notation.}: 
  \begin{enumerate}[{\bf [CDC.1]}]
    \item 
    $\D[f+g] = \D[f] + \D[g]$ and $\D[0]=0$;
    \item 
    $\<a,b+c\>\D[f] = \<a,b\> \D[f]+\<a,c\>\D[f]$ and $\<a,0\>\D[f] = 0$;
    \item 
    $\D[1]=\pi_1$, $\D[\pi_0]=\pi_1\pi_0$, and $\D[\pi_1] =\pi_1\pi_1$;
    \item 
    $\D[\<f,g\>] = \<\D[f],\D[g]\>$;
    \item 
    $\D[fg] = \<\pi_0f,\D[f]\>\D[g]$;
    \item 
    $\<\<a,b\>,\<0,c\>\>\D[\D[f]] = \<a,c\>\D[f]$;
    \item 
    $\<\<a,b\>,\<c,d\>\>\D[\D[f]] = \<\<a,c\>,\<b,d\>\>\D[\D[f]]$.
  \end{enumerate}
\end{definition}

For an in-depth commentary on these axioms, we invite the reader to see the original Cartesian differential category paper  \cite{journal:BCS:CDC}. Briefly, {\bf [CDC.1]} is that the derivative of a sum is the sum of the derivatives, {\bf [CDC.2]} states that derivatives are additive in their second argument, {\bf [CDC.3]} says that the identity and projection maps are linear (more on what this means soon), {\bf [CDC.4]} is that the derivative of a pairing is the pairing of the derivatives, {\bf [CDC.5]} is the famous chain rule, {\bf [CDC.6]} says that the derivative is linear in its second argument, and finally {\bf [CDC.7]} is the symmetry of the mixed partial derivatives. 

\begin{example}\label{example:CDC} Here are some well-known examples of Cartesian differential categories.
\begin{enumerate}
\item Every category with finite biproducts is a Cartesian differential category where for a map $A \to^{f} B$, its derivative $A \oplus A \to^{\D[f]} B$ is defined as $\D[f] := A \oplus A \to^{\pi_1} A \to^f B$. 
\item Let $R$ be a commutative rig. $\mathsf{POLY}_R$ is a Cartesian differential category whose differential combinator is given by the standard differentiation of polynomials. By {\bf [CDC.4]}, since every map in $\mathsf{POLY}_R$ is a tuple, it is sufficient to define the derivative of maps $n \to^p 1$, which are polynomials $p(\vec x) \in R[x_1, \hdots, x_n]$. Then its derivative $n \times n \to^{\D[p]} 1$, viewed as polynomials $\D[p](\vec x, \vec y) \in R[x_1, \hdots, x_n, y_1, \hdots, y_n]$, is defined by the sum of partial derivatives of $p(\vec x)$:
$$ \D[p](\vec x, \vec y) := \sum \limits^n_{i=1} \frac{\partial p}{\partial x_i}(\vec x) y_i $$
For example, consider the polynomial $p(x_1,x_2) = x_1^2 + 3x_1x_2 + 5x_2$, so $2 \to^p 1$, then $4 \to^{\D[p]} 1$ is $ \D[p](x_1, x_2, y_1, y_2) = (2x_1 + 3x_2)y_1 +  (3x_1 + 5)y_2$. On the other hand, for a map $n \to^P m$, which is a tuple $P := \langle p_1(\vec x), \hdots, p_m(\vec x) \rangle$, its derivative is the tuple $\D[P] := \langle \D[p_1](\vec x, \vec y), \hdots, \D[p_m](\vec x, \vec y) \rangle$. 
\item   The category $\smooth$ is a Cartesian differential category where for a map $n \to^F m$, which is a smooth function $\mathbb{R}^n \to^F \mathbb{R}^m$, its derivative $\mathbb{R}^n \x \mathbb{R}^n \to^{\D[F]} \mathbb{R}^m$ is defined as 
$$\D[F](\vec x, \vec v) := J_F(\vec x) \cdot v$$
where $J_F(\vec x)$ is the Jacobian of $F$ at $\vec x$ and where $\cdot$ is matrix multiplication. Of course, similar to the previous example, as every $F$ can be viewed as a tuple, by  {\bf [CDC.4]}, it would have also been sufficient to define the differential combinator for smooth maps $\mathbb{R}^n \to^f \mathbb{R}$. In this case, $J_f(x)$ is better known as the gradient of $f$, $\nabla(f)(x) := \langle \frac{\partial f}{\partial x_1}(\vec x), \hdots, \frac{\partial f}{\partial x_n}(\vec x) \rangle$, and so $\mathbb{R}^n \x \mathbb{R}^n \to^{\D[F]} \mathbb{R}$ is: 
$$ \D[f](\vec x, \vec y) := \nabla(f)(\vec x) \cdot \vec y = \sum \limits^n_{i=1} \frac{\partial f}{\partial x_i}(\vec x) y_i $$
This clearly shows that $\mathsf{POLY}_\mathbb{R}$ is a sub-Cartesian differential category of $\smooth$. 
\end{enumerate}
\end{example}

We now provide a few lemmas that give alternative views on the axioms of a Cartesian differential category; these will be helpful when comparing this structure to a reverse differential category. Note that while the first lemma shows that {\bf [CDC.4]} is actually redundant, to keep the numbering of the equations consistent with past literature on Cartesian differential categories, we chose to 
include it in the definition.

\begin{lemma}\label{lemma:cd4-lemma} \cite[Lemma 2.8]{lemay2018tangent}
  In a Cartesian differential category,
  {\bf [CDC.4]} is redundant.
\end{lemma}

\begin{lemma}\label{lemma:cd6/7-lemma} \cite[Proposition 4.2]{journal:TangentCats} In a Cartesian left additive category:
\begin{enumerate}
\item \label{lemma:cd6-lemma} If a combinator $D$ satisfies {\bf [CDC.1-5,7]}, the axiom {\bf [CDC.6]} is equivalent to:
 $$\<1\x \pi_0,0\x \pi_1\>\D[\D[f]] = (1\x \pi_1)\D[f]$$
  \item \label{lemma:cd7-lemma} If a combinator $D$ satisfies 
  {\bf [CDC.1-6]}, the axiom {\bf [CDC.7]} is equivalent to 
  $$\ex \D[\D[f]] = \D[\D[f]]$$ 
where $\ex: (A\x B) \x (C\x D) \to (A\x C) \x (B\x D)$ is the {\bf exchange} natural isomorphism defined as $\ex := \<\pi_0 \x \pi_0, \pi_1\x \pi_1\>$. 
\end{enumerate}
\end{lemma}





In a Cartesian differential category, there are two important notions: that of partial derivatives and that of linear maps. Beginning with partial derivatives, 
if $A\x B \to^{f} C$ then the partial derivative of $f$ with respect to $B$ is defined as follows: 
$$D_B[f] := A\x (B\x B) \to^{\<1\x \pi_0,0\x \pi_1\>}(A\x B)\x (A\x B)\to^{\D[f]} C $$ 
This partial derivative definition induces a Cartesian differential category on the simple slice categories. Recall that the simple slice category of $\X$ with respect to $A$, denoted $\X[A]$, is the category with the same objects as $\X$
and where a map from $B \to C$ in $\X[A]$ is a map ${f: A \x B \to C}$ in $\X$; that is, in terms of homsets, $\X[A](B,C) = \X(A \x B, C)$, and composition of is given by $\<f,\pi_1\>g$.

\begin{proposition} \cite[Corollary 4.5.2]{journal:BCS:CDC}
  Let $\X$ be a Cartesian differential category and $A$ any object.
  Then $\X[A]$ is a Cartesian differential category and the derivative of 
  $f : A \times B\to C$ is $D_B[f]$.
\end{proposition}

Linear maps play a central role in the theory of Cartesian differential categories. 

\begin{definition}\label{definition:linear}
A map $f$ in a Cartesian differential category is {\bf linear} when $\D[f] = \pi_1f$. Similarly, a map $A\x B \to^{f} C$ is {\bf linear in $B$} if the following diagram commutes:  
$$\begin{tikzcd}
    A\x (B\x B) \ar[dr,"1\x \pi_1"'] \ar[rr,"{D_B[f]}"] && C \\
     & A\x B \ar[ur,"f"']
  \end{tikzcd}  $$ 
\end{definition}

Note that a map $A \x B \to^{f}C$ is linear in $B$ if and only if when regarded as a map ${B \to^{f} C}$ in $\X[A]$, it is linear 
with respect to the derivative in $\X[A]$. 

\begin{example} Let us consider the linear maps in our examples of Cartesian differential categories from Example \ref{example:CDC}: 
\begin{enumerate}
\item In a category with finite biproducts, every map is linear by definition of the differential combinator. 
\item Let $R$ be a commutative rig. In $\mathsf{POLY}_R$, a map $n \to^p 1$ is linear if and only if $p(\vec x) = \sum \limits^n_{i=1} r_i x_i$ for some $r_i \in R$. And it follows that $n \to^{P} m$, with $P = \langle p_1(\vec x), \hdots, p_n(\vec x) \rangle$, is linear if and only if each $p_i(\vec x)$ is. In other words, $n \to^{P} m$ is linear in the Cartesian differential category sense if and only if it induces an $R$-linear map $R^n \to R^m$. 
\item Similar to the previous example, in $\smooth$ the linear maps in the Cartesian differential category sense are precisely the linear maps in the ordinary sense. Explicitly, $n \to^F m$ is linear if and only if $\mathbb{R}^n \to^F \mathbb{R}^m$ is a linear transformation. 
\end{enumerate}
\end{example}

For a Cartesian differential category $\X$, we can also form its subcategory of linear maps $\mathsf{Lin}(\X)$, and since every linear map is additive \cite{journal:BCS:CDC}, it follows that:

\begin{proposition}[\cite{journal:BCS:CDC}, Corollary 2.2.3]\label{prop:linearbiproducts}
For a Cartesian differential category $\X$, its subcategory of linear maps $\mathsf{Lin}(\X)$ has finite biproducts.
\end{proposition}

Finally, we conclude this section with the observation that linearity can also be expressed in terms of injection maps: 

\begin{lemma}\label{lemma:linear-equiv-both} In a Cartesian differential category, 
\begin{itemize}
\item \label{lemma:linear-equiv}
  A map $A \to^{f} B$ is linear if and only if $\iota_1 \D[f] = f$.
  \item \label{lemma:partial-linear-equiv}  A map $A\x B\to^{f} C$ is linear in $B$ if and only if $(\iota_0 \times \iota_1) \D[f] = f$. 
\end{itemize}
\end{lemma}

\section{Reverse derivatives}\label{section:reverse-der}

In this section we introduce our definition of a Cartesian reverse differential category.  The types of axioms are similar to those for Cartesian differential categories; however, after the first two, the forms the axioms take are quite different.  

\begin{definition}
  A Cartesian left additive category $\X$ has {\bf reverse derivatives} in case 
  there is a combinator $\R$, called the {\bf reverse differential combinator}, which written as an inference rule is given by:
 $$   \infer{A\x B\to_{\R[f]} A}{A \to^{f} B}  $$ 
where $\R[f]$ is called the reverse derivative of $f$, and such that the following coherences are satisfied: \\
\noindent 
 {\bf [RD.1]}   $\R[f+g] = \R[f] + \R[g]$ and $\R[0]=0$; \\
\noindent  
{\bf [RD.2]} $\<a,b+c\>\R[f] = \<a,b\>\R[f] + \<a,c\>\R[f]$ and  $\<a,0\>\R[f] = 0$ \\
{\bf [RD.3]} 
    $\R[1] = \pi_1$, while for the projections, the following diagrams commute:
$$ \infer{(A\x B) \x A \to_{\R[\pi_0]} A\x B}{A\x B \to^{\pi_0} A} \quad \quad 
              \begin{tikzcd}
                (A\x B) \x A \ar[dr,"\pi_1"']\ar[rr,"{\R[\pi_0]}"] && A \x B \\
                  & A \ar[ur,"{\iota_0}"']
              \end{tikzcd} $$
$$ \infer{(A\x B) \x B \to_{\R[\pi_1]} A\x B}{A\x B \to^{\pi_1} B} \quad \quad 
              \begin{tikzcd}
                (A\x B) \x B \ar[dr,"\pi_1"']\ar[rr,"{\R[\pi_1]}"] && A \x B \\
                  & B \ar[ur,"{\iota_1}"']
              \end{tikzcd} $$
\noindent {\bf [RD.4]} 
   For a tupling of maps $f$ and $g$, the following equality holds:      
$$ \infer{A\x B \to_{\R[f]} A}{A\to^{f} B} \qquad \infer{A\x C \to_{\R[g]} A}{A\to^{g} C} \qquad \infer{A \x (B\x C) \to_{\R[\<f,g\>]} A}{A \to^{\<f,g\>} B\x C} $$
$$\R[\<f,g\>] = (1\x \pi_0)\R[f]+ (1\x \pi_1)\R[g]$$
While for the unique map to the terminal object: $!_A: A \to \mathbf{1}$, the following equality holds: 
$$\R[!_A] = 0$$
\noindent {\bf [RD.5]} 
For composable maps $f$ and $g$, the following diagram commutes:
$$         \infer{A\x B \to_{\R[f]}A}{A\to^{f}B} 
             \qquad \infer{B\x C \to_{\R[g]} B}{B\to^{g} C}  
             \qquad \infer{A\x C \to_{\R[fg]} A}{A\to^{fg} C}$$
             $$         \begin{tikzcd}[column sep=4em]
              A \x C \ar[r,"{\R[fg]}"] \ar[d,"{\<\pi_0,\<\pi_0f,\pi_1\>\>}"']& A\\
              A\x (B\x C) \ar[r,"{1\x \R[g]}"'] & A\x B \ar[u,"{\R[f]}"']
             \end{tikzcd}$$
\noindent {\bf [RD.6]} 
$ \<1\x \pi_0,0\x \pi_1\>(\iota_0 \x 1)\R[\R[\R[f]]]\pi_1 = (1\x \pi_1)\R[f] $  \\ 
\noindent {\bf [RD.7]} 
$(\iota_0 \x 1)\R[\R[(\iota_0 \x 1)\R[\R[f]]\pi_1]]\pi_1  = \ex (\iota_0 \x 1)\R[\R[(\iota_0 \x 1)\R[\R[f]]\pi_1]]\pi_1$ \\ \\
\noindent A {\bf Cartesian reverse differential category} is a Cartesian
 left additive category with a reverse differential combinator. 
\end{definition}

The axioms of the reverse differential combinator mirror those of a differential combinator. {\bf [RD.1]} states that the reverse derivative of a sum is the sum of the reverse derivatives while {\bf [RD.2]} says that the reverse derivative is additive in its second argument. {\bf [RD.3]} and {\bf [RD.4]} respectively explain what the reverse derivatives of the identity, projection, and tuples are. {\bf [RD.5]} is the reverse derivative version of the chain rule. Lastly, {\bf [RD.6]} expresses that the reverse derivative is linear in its second argument and {\bf [RD.7]} gives the symmetry of the mixed partial reverse derivatives. 

\begin{example}\label{example:rdc} Here are some examples of reverse differential categories: 
\begin{enumerate}
\item Let $R$ be a commutative rig. $\mathsf{POLY}_R$ is a reverse differential category whose reverse differential combinator $\R$ is again defined using partial derivatives of polynomials. For a map $n \to^P m$, $P := \langle p_1(\vec x), \hdots, p_m(\vec x) \rangle$ with $p_i(\vec x) \in R[x_1, \hdots, x_n]$, its reverse derivative $n \times m \to^{\R[P]} n$ is the tuple:
$$\R[P] := \langle \sum \limits^m_{i=1} \frac{\partial p_i}{\partial x_1}(\vec x) y_i , \hdots, \sum \limits^m_{i=1} \frac{\partial p_i}{\partial x_n}(\vec x) y_i \rangle$$
where each component of $\R[P]$ is a polynomial in $R[x_1, \hdots, x_n, y_1, \hdots, y_m]$. For example, consider from before the polynomial $p(x_1,x_2) = x_1^2 + 3x_1x_2 + 5x_2$, then ${3 \to^{\R[p]} 2}$ is the tuple of polynomials in $3$ variables, $\R[p] = \langle (2x_1 + 3x_2)y, (3x_1 + 5)y \rangle$.
\item $\smooth$ is a reverse differential category whose reverse differential combinator is defined using the transpose of the Jacobian. For a map $n \to^F m$, that is, a smooth function ${\mathbb{R}^n \to^F \mathbb{R}^m}$, its reverse derivative $n \times m \to^{\R[F]} n$ is the smooth map $\mathbb{R}^n \times \mathbb{R}^m \to^{\R[F]} \mathbb{R}^n$ defined as:
$$\R[F](\vec x, \vec y) := (J_f(x))^T \cdot \vec y$$ 
In particular for a smooth map $\mathbb{R}^n \to^f \mathbb{R}$, its reverse derivative $\mathbb{R}^n \times \mathbb{R}  \to^{\R[f]} \mathbb{R}^n$ is calculated out to be:
$$\R[f](\vec x, y) := \langle \frac{\partial f}{\partial x_1}(\vec x)y, \hdots, \frac{\partial f}{\partial x_n}(\vec x)y \rangle$$
And as before, $\mathsf{POLY}_\mathbb{R}$ is a sub-reverse differential category of $\smooth$. 
\end{enumerate}
\end{example}

The following lemma captures some basic properties of the reverse derivative.
\begin{lemma}\label{lemma:reverse-der-basic} In a Cartesian reverse differential category, the following equalities holds: 
  \begin{enumerate}
    \item $\R[fg] = \<\pi_0,\<\pi_0f,\pi_1\>\R[g]\>\R[f]$;
    \item $\R[\iota_0] = \pi_1\pi_0$ and $\R[\iota_1]=\pi_1\pi_1$;
    \item $\R[\pi_0f] = (\pi_0 \x1)\R[f]\iota_0$ and $\R[\pi_1f]=(\pi_1 \x 1)\R[f]\iota_1$;
    \item $\R[f\pi_0] = (1\x \iota_0)\R[f]$ and $\R[f\pi_1] = (1\x \iota_1)\R[f]$;
    \item $\R[f\x g] = \ex (\R[f] \x \R[g])$;
    \item $\R[\iota_0 f] = (\iota_0 \x 1)\R[f] \pi_0$ and $\R[\iota_1 f] = (\iota_1\x 1)\R[f] \pi_1$;
    \item $\R[f\iota_0] = (1\x \pi_0)\R[f] $ and $\R[f\iota_1] = (1\x \pi_1)\R[f]$;
    \item $\R[\<f|g\>] = \<\D[f\iota_0]|\R[g\iota_1]\>$;
    \item $\R[f\oplus g] = \ex(\R[f] \x \R[g])$;
  \end{enumerate}
\end{lemma}
\begin{proof}
    We have the following calculations.
    \begin{enumerate}
      \item Immediate.  
      \item $\R[\iota_0] = \R[\<1,0\>] = (1\x \pi_0)\D[1] = (1\x \pi_1)\D[0] = (1\x \pi_0)\pi_1 = \pi_1\pi_0$.
            Similarly $\R[\iota_1] = \pi_1\pi_1$.
      \item We have 
            \begin{align*}
              \R[\pi_0 f] &= \<\pi_0,\<\pi_0\pi_0,\pi_1\>\R[f]\>\R[\pi_0] \\
                          &= \<\pi_0,\<\pi_0\pi_0,\pi_1\>\R[f]\>\pi_1 \iota_0\\
                          &= \<\pi_0,\pi_0,\pi_1\>\R[f]\iota_0.
            \end{align*}
            Similarly, $\R[\pi_1f] = \<\pi_0\pi_1,\pi_1\>\R[f]\iota_1$.
      \item We have 
            \begin{align*}
              \R[f\pi_0] &= \<\pi_0,\<\pi_0f,\pi_1\>\R[\pi_0]\>\R[f] \\
                         &= \<\pi_0,\<\pi_0f,\pi_1\>\pi_1 \iota_0\>\R[f]\\
                         &= \<\pi_0,\pi_1\iota_0\>\R[f] = (1\x \iota_0)\R[f]
            \end{align*}
            Similarly, $\R[f\pi_1] = (1\x \iota_1)\R[f]$.
      \item We have 
            \begin{align*}
               \R[f\x g] &= \R[\<\pi_0f,\pi_1g\>] \\
                       &= (1\x \pi_0)\R[\pi_0 f] + (1\x \pi_1)\R[\pi_1 g]\\
                       &= (1\x \pi_0)(\pi_0 \x 1)\R[f]\iota_0 + (1\x \pi_1)(\pi_1 \x 1)\R[g]\iota_1\\
                       &= \<(\pi_0 \x \pi_0)\R[f],(\pi_1\x \pi_1)\R[g]\> \\
                       &= \ex \<\pi_0 \R[f],\pi_1 \R[g]\> = \ex (\R[f] + \R[g]).
            \end{align*}
      \item We have 
            \begin{align*}
              \R[\iota_0 f]
                &= \<\pi_0,\<\pi_0\iota_0,\pi_1\>\R[f]\>\R[\iota_0] \\
                &= \<\pi_0,\<\pi_0,\iota_0,\pi_1\>\R[f]\>\pi_1\pi_0 \\
                &= \<\pi_0\iota_0,\pi_1\>\R[f]\pi_0 = (\iota_0 \x 1)\R[f]\pi_0
            \end{align*}
            Similarly, $\R[\iota_1f] = (\iota_1 \x 1)f \pi_1$.
      \item We have $\D[f\iota_0] = \D[\<f,0\>] = (1\x \pi_0) \R[f] + 0 = (1\x \pi_0)\R[f]$.  Similarly, $\R[f\iota_1] = (1\x \pi_1)\R[f]$.
      \item We have 
            \begin{align*}
              \R[\<f|g\>] &= \R[\pi_0f + \pi_1 g]\\
                          &= \R[\pi_0f] + \R[\pi_1g] \\
                          &= (\pi_0 \x 1)\R[f]\iota_0 + (\pi_1\x 1)\R[g]\iota_1\\
                          &= \<(\pi_0 \x 1)\R[f],(\pi_1\x 1)\R[g]\> \\
                          &= \<\R[f\iota_0],\R[g\iota_1]\>
            \end{align*}
      \item Immediate.
    \end{enumerate}
  \end{proof}

\subsection{Forward Differential Structure} 

Here we explain how every reverse derivative operator induces a forward derivative 
operator, that is, how every Cartesian reverse differential category is a Cartesian 
differential category.  The trick was noticed in \cite{chr12}: the
reverse derivative in $\smooth$ is the transpose of the Jacobian, which is linear,
hence applying the reverse derivative again allows one to reconstruct the forward derivative.
We formalize this in an arbitrary Cartesian reverse differential category 
as follows.  Consider the resulting type of applying the reverse 
differential combinator twice: 
 $$   \infer{
      \infer{
        (A\x B)\x A \to_{\R[\R[f]]} (A\x B)
      }{A\x B \to^{\R[f]} A}
    }{A \to^{f} B} $$ 

\begin{theorem}\label{theorem:cdc-from-rdc} If $\X$ is a Cartesian reverse differential category, then $\X$ is a Cartesian differential category with differential combinator $\D$ defined as follows (for any map ${A \to^f B}$): 
$$ \D[f] := A \times A \to^{(\<1,0\> \x 1)} (A \times B) \times A \to^{\R[\R[f]]} A \times B \to^{\pi_1} B $$
\end{theorem}
\begin{proof}
    We will show all of the axioms for a Cartesian differential category 
    hold. 
    
    \noindent {\bf [CDC.1]}    \begin{align*}
          \D[f+g] &= (\<1,0\>\x 1)\R[\R[f+g]]  \pi_1 \\
                 &= (\<1,0\>\x 1)(\R[\R[f]] + \R[\R[g]])\pi_1\\
                 &=(\<1,0\>\x 1)(\R[\R[f]]\pi_1 + (\<1,0\>\x 1)\R[\R[g]]\pi_1  \\
                 &= \D[f] + \D[g]            
        \end{align*}
        Similarly, $\D[0] = 0$.
  
    \noindent {\bf [CDC.2]}  
        \begin{align*}
          \<a,b+c\>\D[f] &= \<a,b+c\>(\<1,0\> \x 1)\R[\R[f]]\pi_1 \\
                        &= \<a\<1,0\>,b+c\>\R[\R[f]]\pi_1\\
                        &= \<a\<1,0\>,b\>\R[\R[f]]\pi_1 + \<a\<1,0\>,c\>\R[\R[f]]\pi_1\\
                        &= \<a,b\>\R[\R[f]] + \<a,c\>\R[\R[f]]
        \end{align*}
        Similarly, $\<a,0\>\D[f] = 0$.
  
    \noindent {\bf [CDC.3]}  
        \begin{align*}
          \D[1] &= (\<1,0\>\x 1)\R[\R[1]]\pi_1\\
               &= (\<1,0\> \x 1)\R[\pi_1]\pi_1\\
               &= (\<1,0\> \x 1)\pi_1\<0,1\>\pi_1\\
               &= \pi_1
        \end{align*}
        \begin{align*}
          \D[\pi_1] 
            &= (\<1,0\>\x 1)\R[\R[\pi_1]]\pi_1 \\
            &= (\<1,0\> \x 1)\R[\pi_1\iota_1]\pi_1\\
            &= (\<1,0\> \x 1)(\pi_1\x \pi_1) \D[\iota_1]\iota_1\pi_1\\
            &= \pi_1\pi_1\iota_1\pi_1 = \pi_1\pi_1
        \end{align*}
        Similarly, $\D[\pi_0] = \pi_1\pi_0$.
  
    \noindent {\bf [CDC.4]}   Immediate from Lemma \ref{lemma:cd4-lemma}.  
    
    \noindent {\bf [CDC.5]}  
        Our goal is to show that $\D[fg] = \<\pi_0f,\D[f]\>\D[g]$.
        First, consider 
        \begin{align*}
           \<\pi_0f,\D[f]\>\D[g] 
           &= \<\pi_0f,(\<1,0\> \x 1)\RR[f]\pi_1\>(\<1,0\> \x 1)\RR[g]\pi_1 \\
           &=  \<\pi_0 f\<1,0\>,(\<1,0\> \x 1)\RR[f]\pi_1\>\RR[g]\pi_1
        \end{align*}
        Next:
        \begin{align*}
          & \D[fg] \\
          &= (\<1,0\> \x 1)\RR[fg]\pi_1 \\
          &= (\<1,0\> \x 1)\R[\<\pi_0,\<\pi_0f,\pi_1\>\R[g]\>\R[f]]\pi_1\\
          &= (\<1,0\> \x 1) \<\pi_0,\<\pi_0\<\pi_0,\<\pi_0f,\pi_1\>\R[g]\>,\pi_1\>\RR[f]\> q \\
          &= \<\<\pi_0,0\>,\<\<\pi_0,0\>\<\pi_0,\<\pi_0f,\pi_1\>\R[g]\>,\pi_1\>\RR[f]\>q \\
          &= \<\<\pi_0,0\>,\<\<\pi_0,\<\pi_0,0\>\<\pi_0f,\pi_1\>\R[g]\>,\pi_1\>\RR[f]\>q \\
          &= \<\<\pi_0,0\>,\<\<\pi_0,\<\pi_0f,0\>\R[g]\>,\pi_1\>\RR[f]\>q\\
          &= \<\<\pi_0,0\>,\<\<\pi_0,0\>,\pi_1\>\RR[f]\>q \tag*{\bf [RD.2]} 
        \end{align*}
        where $q = \R[\<\pi_0,\<\pi_0f,\pi_1\>\R[g]\>]\pi_1$.
        Next we simplify $q$.
        \begin{align*}
          & q 
          = \R[\<\pi_0,\<\pi_0f,\pi_1\>\R[g]\>]\pi_1 \\
          &= \left((1\x \pi_0)\R[\pi_0] + (1\x \pi_1)\R[\<\pi_0f,\pi_1\>\R[g]] \right)\pi_1 \tag*{\bf [RD.4]}\\
          &= ((1\x \pi_0)\pi_1\iota_1 + 
              (1\x \pi_1)\<\pi_0,\<\pi_0\<\pi_0f,\pi_1\>,\pi_1\>\RR[g]\>\R[\<\pi_0f,\pi_1\>])\pi_1 \tag*{\bf [RD.5]}\\
          &= (\pi_1\pi_0\iota_0
          + 
          (1\x \pi_1)\<\pi_0,\<\pi_0\<\pi_0f,\pi_1\>,\pi_1\>\RR[g]\>
          ( (1\x \pi_0)\R[\pi_0 f] + (1\x \pi_1)\R[\pi_1] ))\pi_1\\
          &= \pi_1\pi_0\iota_0\pi_1\\
          &+ 
          (1\x \pi_1)\<\pi_0,\<\pi_0\<\pi_0f,\pi_1\>,\pi_1\>\RR[g]\>
          ( (1\x \pi_0)(\pi_0 \x 1)\R[f]\iota_0 + (1\x \pi_1)\pi_1\iota_1  ) \pi_1\\
          &= 0 +
          (1\x \pi_1)\<\pi_0,\<\pi_0\<\pi_0f,\pi_1\>,\pi_1\>\RR[g]\>
          ( \<(\pi_0\x \pi_0)\RR[f], \pi_1\pi_1\> )\pi_1 \\
          &=
          (1\x \pi_1)\<\pi_0,\<\pi_0\<\pi_0f,\pi_1\>,\pi_1\>\RR[g]\>
          \pi_1\pi_1\\
          &= (1\x \pi_1)\<\pi_0\<\pi_0f,\pi_1\>,\pi_1\>\RR[g]\pi_1
        \end{align*}
        Then we plug $q$ back into the  formula for $\D[fg]$ and continue simplifying.
        \begin{align*}
          & \<\<\pi_0,0\>,\<\<\pi_0,0\>,\pi_1\>\RR[f]\>q \\
          &= \<\<\pi_0,0\>,\<\<\pi_0,0\>,\pi_1\>\RR[f]\> (1\x \pi_1)\<\pi_0\<\pi_0f,\pi_1\>,\pi_1\>\RR[g]\pi_1 \\
          &= \<\<\pi_0,0\>,\<\<\pi_0,0\>,\pi_1\>\RR[f]\pi_1\>\<\pi_0\<\pi_0f,\pi_1\>,\pi_1\>\RR[g]\pi_1\\
          &= \<\<\pi_0,0\>\<\pi_0f,\pi_1\>,\<\<\pi_0,0\>,\pi_1\>\RR[f]\pi_1\>\RR[g]\pi_1\\
          &= \<\<\pi_0f,0\>,(\<1,0\> \x 1)\RR[f]\pi_1\>\RR[g]\pi_1\\
          &= \<\pi_0f,\D[f]\>\D[g]
        \end{align*}
        as desired.
  
    \noindent {\bf [CDC.6]}  Note that with the definition of the forward derivative introduced here,
          {\bf [RD.6]} is the same as $\<1\x \pi_0,0\x \pi_1\>\D[\R[f]] = (1\x \pi_1)\R[f]$. First, we will show a more general claim than needed for this point.
          Note that we have already shown that {\bf [CDC.5]} holds: thus, we have
          that if $\D[g]=\pi_1 g$ and $\D[k] = \pi_1 k$ then:
          \[
            \D[gfk] = (g\x g) \D[f]k  
          \]
          for any $f$.  The proof is straightforward:
          \begin{align*}
            \D[gfk] &= \<\pi_0 g,\D[g]\>\D[fk] \\
                   &= \<\pi_0 g,\D[g]\>\<\pi_0f,\D[f]\>\D[k]\\
                   &= \<\pi_0 g,\pi_1 g\>\<\pi_0f,\D[f]\>\pi_1 k\\
                   &= (g\x g)\D[f] k
          \end{align*}
  We have also shown that $\D[\pi_1] = \pi_1\pi_1$. 
  
  Suppose $A\x B \to^{h} C$.  Then note the types 
          $(A\x B) \x C \to^{\R[h]} A\x B$ and
          \[
            A\x C \to^{\<1,0\> \x 1} (A\x B) \x C \to^{\R[h]} A\x B \to^{\pi_1} B  
          \]
          
          We will show that the following diagram always
          commutes:
  
          \[
            \begin{tikzcd}[column sep=5em]
              A \x (C\x C) \ar[dr,"1\x \pi_1"']\ar[rr,"{D_C[(\<1,0\>\x 1)\R[h]\pi_1]}"] && B \\
               & A\x C \ar[ur,"{(\<1,0\> \x 1)\R[h]\pi_1}"']
            \end{tikzcd}  
          \]

          Note in the above we are using $D_C[F]$ as shorthand for $\<1\x\pi_0,0\x \pi_1\>\D[F]$.
          Then
          \begin{align*}
            &  D_C[(\<1,0\> \x 1)\R[h]\pi_1]\\
            &= \<1\x \pi_0,0\x \pi_1\>\D[(\<1,0\> \x 1)\R[h]\pi_1]\\
            &= \<1\x \pi_0,0\x \pi_1\>((\<1,0\> \x 1)\x (\<1,0\> \x 1))\D[\R[h]]\pi_1 \\
            &= \< (1\x \pi_0)(\<1,0\> \x 1),(0\x\pi_1)(\<1,0\>\x 1)\>\D[\R[h]]\pi_1 \\
            &= \< (\<1,0\> \x \pi_0), (0\x \pi_1)\>\D[\R[h]]\pi_1 \\
            &= \< (\<1,0\> \x 1)(1\x \pi_0),(\<1,0\> \x 1)(0\x \pi_1)\>\D[\R[h]]\pi_1 \qquad 0 = \<1,0\> 0\\
            &= (\<1,0\> \x 1)\<1\x \pi_0,0\x \pi_1\>\D[\R[h]]\pi_1 \\
            &= (\<1,0\> \x 1)(1\x \pi_1)\R[h]\pi_1 \qquad \text{\bf [RD.6]}\\
            &= (1\x \pi_1)(\<1,0\> \x 1)\R[h]\pi_1
          \end{align*}
  
          Then note that letting $A \to^{f} B$, and setting $h = \R[f]$ in the above
          formula, we have 
          \[
            \<1\x \pi_0,0\x \pi_1\>\D[(\<1,0\> \x 1)\R[\R[f]]\pi_1] = (1\x \pi_1)  (\<1,0\> \x 1)\R[\R[f]]\pi_1
          \]
          But then by definition we have 
          \[
            \<1\x \pi_0,0\x \pi_1\>\D[\D[f]] = (1\x \pi_1)\D[f]  
          \]
          which is {\bf [CDC.6]}.
  
    \noindent {\bf [CDC.1]}   {\bf [CDC.7]}  With the definition of the forward derivative in mind, {\bf [RD.7]}
          may be re-expressed as $\ex \D[\D[f]] = \D[\D[f]]$.  Then use 
          Lemma \ref{lemma:cd6/7-lemma} to conclude that {\bf [CDC.7]} holds.  \noindent {\bf [CDC.1]}  
  \end{proof}

\begin{example} For both $\mathsf{POLY}_\mathbb{R}$ and $\smooth$, applying Theorem \ref{theorem:cdc-from-rdc} to their respective reverse differential operators defined in Example \ref{example:rdc} results precisely in their differential combinators defined in Example \ref{example:CDC}. This follows from the fact that there is a bijective correspondence between a reverse differential combinator and a differential combinator with an involution operation, which we will discuss in Section \ref{section:cont-dagger}. 
\end{example}

\subsection{Dagger Structure and Linear Maps}\label{section:dagger-cats}

We now investigate the subcategory of linear maps of the induced Cartesian differential category structure from Theorem \ref{theorem:cdc-from-rdc} of a Cartesian reverse differential category. In particular we will show that the subcategory of linear maps has a dagger structure. 

\begin{definition} A {\bf $\dagger$-category} \cite{journal:selinger-dagger} is a category $\X$ with a stationary on objects 
involution $\X^\text{op} \to^{(\blank)^\dagger} \X$. A $\dagger$-category that also has finite biproducts $\oplus$, with projection maps $\pi_i$ and injection maps $\iota_i$, is said to have {\bf $\dagger$-biproducts} \cite{journal:selinger-dagger} when $\pi_i^\dagger = \iota_i$ (or equivalently if $\iota^\dagger_i = \pi_i$).  \end{definition}

Note that having $\dagger$-biproducts implies that $0^\dagger = 0$ and $(f+g)^\dagger = f^\dagger + g^\dagger$. At this point we can also point out that in the same way that every category with finite biproducts is a Cartesian differential category, we have the following basic example of a reverse differential category: 

\begin{example}
Every $\dagger$-category with finite $\dagger$-biproducts is a reverse differential category where for a map $A \to^f B$, $A \oplus B \to^{\R[f]} A$ is defined as $\R[f] := A \oplus B \to^{\pi_1} B \to^{f^\dagger} A$. As a particular example, let $R$ be a commutative rig and let $\mathsf{MAT}(R)$ be the category of matrices over $R$, that is, the category whose objects are the natural numbers $n \in \mathbb{N}$ and where a map $n \to^A m$ is an $n \times m$-matrix $A$ with coefficients in $R$. $\mathsf{MAT}(R)$ admits finite biproducts where on objects $n \oplus m := n + m$ and where the projection and injection maps are the obvious matrices. $\mathsf{MAT}(R)$ also admits a $\dagger$ defined as the transpose of matrices and this makes $\mathsf{MAT}(R)$ into a $\dagger$-category with finite $\dagger$-biproducts. 
\end{example}

For any map $A \to^f B$ in a reverse differential category, we can define a map of opposite type $B \to^{f^\dagger} A$ by $f^\dagger := \iota_1\R[f]$.  As the following example shows, however, in general this operation is neither functorial nor involutive.   
\begin{example} With our standard example $2 \to^p 1$  in $\mathsf{POLY}_R$, $p(x_1,x_2) = x_1^2 + 3x_1x_2 + 5x_2$, one computes that $1 \to^{p^\dagger} 2$ is the tuple of $1$ variable polynomials $p^\dagger = \langle 0, 5x \rangle$.  \end{example}

However, as we shall see, $\dagger$ is well behaved for linear maps.  

\begin{lemma}\label{lemma:reverse-dagger-biprod-coherences}
  With the preceding definition of $\dagger$ in a reverse differential category,
 one has that $\pi_i^\dagger = \iota_i$ and $\iota_i^\dagger = \pi_i$. 
\end{lemma}

\begin{lemma}\label{lemma:linear-reverse}
  In a Cartesian reverse differential category, for any map $A \to^{f} B$, the 
  following are equivalent:
  \begin{enumerate}
    \item $f$ is linear (Definition \ref{definition:linear}) with respect to the differential combinator of Theorem \ref{theorem:cdc-from-rdc};
    \item $\iota_1(\iota_0 \x 1)\R[\R[f]] \pi_1 = f$;
    \item $f^{\dagger\dagger} = f$.
  \end{enumerate}
\end{lemma}
\begin{proof} That \emph{1} $\Leftrightarrow$ \emph{2} follows from the fact that by definition, the left hand side of \emph{2} can be re-expressed as $\iota_1(\iota_0 \x 1)\R[\R[f]] \pi_1 = \<0,1\>\D[f]$, and so \emph{2} holds precisely when $\<0,1\>\D[f] = f$, which by Lemma \ref{lemma:linear-equiv} is equivalent to $\D[f]=\pi_1 f$, that is, that $f$ is linear. Next we show that \emph{2} $\Leftrightarrow$ \emph{3}. First note that $\iota_1 (\iota_0 \times 1) = \iota_1 (\iota_1 \times 1)$ since: 
$$\iota_1 (\iota_0 \times 1) = \<0,1\>(\<1,0\> \x 1) = \<0,1\>
    = \<0,1\>(\<0,1\> \x 1) = \iota_1 (\iota_1 \times 1)  $$
And then by Lemma \ref{lemma:reverse-der-basic}.6, we have the following equality: 
$$ f^{\dagger\dagger} = \iota_1 \R[\iota_1 \R[f]] = \iota_1 (\iota_1 \x 1)\R[\R[f]] \pi_1= \iota_1 (\iota_0 \x 1) \R[\R[f]]\pi_1 $$
Then it immediately follows that $f^{\dagger\dagger}=f$ if and only if $f = \iota_1(\iota_0\x 1)\R[\R[f]]\pi_1$.
\end{proof}

\begin{lemma}\label{lemma:rd6-alternative-form}
  In a Cartesian reverse differential category, for any $A \to^{f} B$, its reverse derivative $A\x B \to^{\R[f]}A$
  is linear in $B$ (Definition \ref{definition:linear}) with respect to the differential combinator of Theorem \ref{theorem:cdc-from-rdc}. Furthermore, the following diagram commutes: 
$$        \begin{tikzcd}[column sep=4em]
              ((A\x B)\x A)\x (A\x B) \ar[r,"{\R^{(3)}[f]}"] & (A\x B) \x A \ar[d,"\pi_1"] \\
              A\x B \ar[r,"{\R[f]}"'] \ar[u,"{\<\iota_0,0\> \x \iota_1}"] & A
            \end{tikzcd}  $$
\end{lemma}
\begin{proof}
That $A\x B \to^{\R[f]}A$ is linear in $B$ follows immediately from the {\bf[RD.6]} (we leave it as an exercise to re-express {\bf[RD.6]} in terms of partial derivatives). Commutativity of the diagram follows by applying Lemma \ref{lemma:partial-linear-equiv} to $\R[f]$.
\end{proof}


\begin{proposition}\label{prop:reverseGivesDaggerBiproducts}
 For a Cartesian reverse differential category $\X$, the category of linear maps of the induced Cartesian differential category structure from Theorem \ref{theorem:cdc-from-rdc}, $\mathsf{Lin}(\X)$, is a $\dagger$-category with finite $\dagger$-biproducts. 
\end{proposition}
\begin{proof} By Proposition \ref{prop:linearbiproducts}, we already know that 
  $\mathsf{Lin}(\X)$ has finite biproducts. We need to show that $\mathsf{Lin}(\X)$ 
  also has a $\dagger$. Lemma \ref{lemma:linear-reverse} shows that the linear maps
   are precisely those for which $f^{\dagger\dagger} = f$, and thus if $f$ is linear then 
   $f^\dagger$ is linear. Therefore $\dagger$ is well-defined and involutive. We now show that $\dagger$ is a contravariant functor. First that $\dagger$ preserves the identity:
$$ 1^\dagger = \iota_1 \R[1] = \iota_1 \pi_1 = 1$$
Next, that $\dagger$ preserves composition (recall that if $f$ is linear, then $0f = 0$):
   $$ (fg)^\dagger = \iota_1 \R[fg] = \iota_1 \<\pi_0,\<\pi_0f,\pi_1\>\R[g]\>\R[f] = \<\iota_1\pi_0,\<\iota_1\pi_0f,\iota_1\pi_1\>\R[g]\>\R[f]$$
    $$= \<0,\<0f,1\>\R[g]\>\R[f] = \<0,\<0,1\>\R[g]\>\R[f] = \<0,1\>\R[g]\<0,1\>\R[f] = g^\dagger f^\dagger $$ 
  Note in the above that functoriality only relies on $f$ preserving $0$. Thus $\mathsf{Lin}(\X)$ is a $\dagger$-category. Lastly by Lemma \ref{lemma:reverse-dagger-biprod-coherences}, $\mathsf{Lin}(\X)$ also has $\dagger$-biproducts. 
\end{proof}

\section{From forward derivatives to reverse derivatives}\label{section:cont-dagger}
In the previous section, we showed that a Cartesian reverse differential category gives 
rise to a Cartesian differential category in which the subcategory of linear maps has a dagger biproduct structure. For the converse we need  to develop Cartesian differential categories where every simple slice linear map category is a dagger 
category with dagger biproducts. The conceptual structure behind this is what 
we call a dagger fibration with fibered dagger biproducts.  We will show that 
when a Cartesian differential category's linear map fibration is such a 
dagger fibration then the category is also a Cartesian reverse differential category.

\subsection{Review of Fibrations and the Dual Fibration}
We first recall the notion of fibration (for example, see \cite[Section 1.1]{book:Jacobs-Cat-Log}) and the lesser-known idea of the dual of a fibration.  These will be helpful concepts in which to frame our characterization of reverse differential categories (Theorem \ref{thm:reverseDerivEquiv}) and to describe how the reverse derivative is functorial (Proposition \ref{proposition:reverse-functor}).  
\begin{definition}
Suppose that $q: \X \to \mathbb{B}$ is a functor.  
\begin{enumerate}	
	\item Say that a morphism $f: X \to Y$ in $\X$ is \textbf{over} a morphism $u: I \to J$ in $\mathbb{B}$ if $q(f) = u$.  
	\item Say that a morphism $f: X \to Y$ in $\X$ is \textbf{Cartesian over $u: I \to J$ in $\mathbb{B}$} if $f$ is over $u$, and for every $g: Z \to Y$ in $\X$ such that $q(q) = wu$ for some $w: q(Z) \to I$, there is a unique $h: Z \to X$ in $\X$ over $w$ such that $hf = g$:
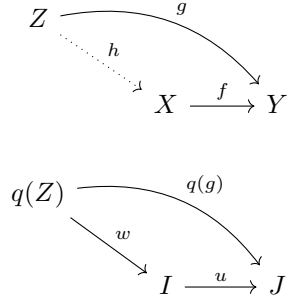
	$$
	\begin{tikzcd}
	Z \ar[dr, dotted, "h"] \ar[drr, bend left, "g"] & & \\
	& X \ar[r,"f"] & Y \\
	q(Z) \ar[dr, "w"] \ar[drr, bend left, "q(g)"] & & \\
	& I \ar[r,"u"] & J \\	
	\end{tikzcd}
	$$
	\item Say that $q$ is a \textbf{fibration} if for every $Y$ in $\X$ and every $u: I \to q(Y)$ in $\mathbb{B}$ , there is a Cartesian morphism $f: X \to Y$ in $\X$ above $u$.  
	\item Say that an arrow $f: X \to Y$ in $\X$ is \textbf{vertical} if $f$ is over an identity map.  
	\item For an object $I$ in $\mathbb{B}$, the \textbf{fibre of $q$ over $I$}, denoted $q^{-1}(I)$, is the category whose objects are those objects of $\X$ for which $q(X) = I$, and whose arrows are vertical morphisms between them.  
\end{enumerate}
\end{definition}

\begin{example}\label{ex:fibration} If $\X$ is a Cartesian category, then the \textbf{simple fibration} \cite[Definition 1.3.1]{book:Jacobs-Cat-Log} $\widetilde{\X} \to^{\pi} \X$ is described as follows:  the total category $\widetilde{\X}$ has objects pairs of objects of $\X$ and a map  $(I,A) \to^{(f,g)} (J,B)$ is given by a pair of maps of type $I \to^{f} J$ and $I \x A \to^{g} B$.  The identity of $(I,A)$ is $(1_A, \pi_1)$ while the composition of maps $(I,A) \to^{(f,g)} (J,B)$ and $(J,B) \to^{(f',g')} (K,C)$ is defined as: $(I\to^{f} J \to^{f^\prime} K,  I\x A \to^{\<\pi_0f,g\>} J \x B \to^{g'} C)$. 
The fibration $\widetilde{\X} \to^{\pi} \X$ is the functor which on objects is $\pi(I,A) = I$ and on maps is $\pi(f,g) := f$. The 
vertical arrows in $\widetilde{\X}$ are precisely those of the form $(I,A) \to^{(1,g)} (I,B)$ while the Cartesian arrows are those of the form $ (I,A) \to^{(f,\pi_1)} (J,A)$. 
\end{example}

\begin{example}
If $\X$ is a Cartesian differential category, we denote by $\widetilde{\mathsf{Lin}(\X)}$ 
the \textbf{simple linear fibration}, whose objects are pairs of objects in $\X$ 
and whose maps $(I,A) \to^{f,g} (J,B)$ are pairs of maps $I \to^{f} J$ and $I \x A \to^{g} B$ 
where $g$ is linear in $B$. Composition and identities of $\widetilde{\mathsf{Lin}(\X)}$ are defined as for the simple fibration. The fiber over $A$ of this fibration is denoted $\mathsf{Lin}(\X)[A]$.  Note that by \cite[Proposition 1.5.4]{journal:BCS:CDC}, every fiber of $\widetilde{\mathsf{Lin}(\X)}$ 
has biproducts.
\end{example}

\begin{definition} 
Suppose that $\X \to^{q} \mathbb{B}$ is a fibration.  The {\bf dual fibration} of $q$
\cite{book:categoricalalgebra,arxiv:simple-dual} 
is a fibration $\X^{\ast} \to^{q^{\ast}} \mathbb{B}$ whose total category $\X^\ast$ has the same objects as $\X$ and where a map $X \to Y$ in $\X^\ast$ is an equivalence class of spans 
$$               \begin{tikzcd}
                    S \ar[r,"c"] \ar[d,"v"'] & Y \\
                    X
                   \end{tikzcd}  $$
                where $v$ is vertical and $c$ is Cartesian (over $q(c)$) under 
                the equivalence relation $(v,c) \sim (v',c')$ when there is 
                a vertical isomorphism $\alpha$ that makes the following 
                diagram commute.$$
                    \begin{tikzcd}
                        & S' \ar[d,"\alpha"]\ar[dl,"v'"'] \ar[dr,"c'"] \\
                        X & S \ar[l,"v"] \ar[r,"c"'] & Y
                    \end{tikzcd}  $$
To compose such spans, note that given a cospan $S \to^{c} X' \from^{v'} S'$ with $c$ cartesian and $v'$ vertical,that there is a cartesian arrow over $q(c)$ with codomain $S'$, and this induces uniquely  a $v''$ making the relevant square commute, and we get a  
span $S \from^{v''} S'' \to^{\hat{c}} S'$ with $v''$ vertical and $\hat{c}$ cartesian; this span is used to from the composite of the spans $(v,c) (v',c')$.
For more details, see \cite{arxiv:simple-dual}. The fibration $q^\ast$ is defined 
on objects as $q^\ast(A) := q(A)$, and defined on maps as $q^\ast(v,c) := q(X)=q(S) \to^{q(C)} q(Y)$.
\end{definition}

\begin{example}The dual of the simple fibration, $\widetilde{X}^\ast$, can be described as the category with objects pairs of objects of $\X$ and with maps $(I,A) \to^{(f,g)} (J,B)$ where $I\to^{f} J$ and a $I\x B \to^{g} A$. The identity on $(I,A)$ is $(1,\pi_1)$, while composition of maps $(I,A) \to^{(f,g)} (J,B)$ and ${(J,B) \to^{(f',g')} (K,C)}$ is 
                  defined to be 
$$ (I \to^{ff'} K, I \x C \to^{\<\pi_0,\<\pi_0f,\pi_1\>\>} (I \x (J \x C)) \to^{1\x g'} I \x B \to^{g} A).  $$ 
\end{example}

\begin{example}
The dual of the linear fibration, $\widetilde{\mathsf{Lin}(\X)}^\ast$, has again objects $(I,A)$
but now maps $(I,A) \to^{(f,g)} (J,B)$ consist of pairs of a map $I \to^{f} J$ and a map 
$I\x B \to^{g} A$ such that $g$ is linear in $B$. 
\end{example}

The dual of the linear fibration allows us to describe how the reverse derivative is functorial:
\begin{proposition}\label{proposition:reverse-functor}
  For a Cartesian reverse differential category $\X$, there is a product-preserving functor ${\X \to \widetilde{\mathsf{Lin}(\X)}^\ast}$ defined on objects as $A \mapsto (A,A)$ and on maps as $f \mapsto (f,R[f])$. 
\end{proposition}
\begin{proof}
  This follows from {\bf [RD.3]} and {\bf [RD.5]}.  
\end{proof}

\begin{lemma}\label{lemma:dual-fiber-lemma}
  A fiber of the dual fibration is isomorphic to the opposite category of the associated fiber of the starting 
  fibration; that is, for any $A$ in $\mathbb{B}$, ${q^\ast}^{\text{-}1}(A) \simeq (q^{\text{-}1}(A))^\text{op}$ 
  and moreover the isomorphism is stationary on objects.
\end{lemma}
\begin{proof}
    First, $(q^{\text{-}1}(A))^\text{op}$ has 
    \begin{description}
      \item[Obj: ] $X \in \X$ such that $q(X) = A$.
      \item[Arr: ] $X \to^{f} Y$ is a map $Y\to^{f} X$ in $\X$ such that $q(f)=1_A$.
    \end{description}
  
    On the other hand, ${q^\ast}^{\text{-}1}(A)$ has 
    \begin{description}
      \item[Obj: ] $X \in \X$ such that $q(X) = A$. 
      \item[Arr: ] $X \to Y$ are spans 
      \[
        \begin{tikzcd}
         S \ar[r,"c"] \ar[d,"v"'] & Y \\
         X
        \end{tikzcd}  
      \] 
      where $c$ is vertical i.e. $q(c) =1_A$. 
    \end{description}
    Since $c$ is both vertical and Cartesian, there is a unique vertical isomorphism $w$
    that inverts $c$.  Then the span is equivalent to $(wv,1_Y)$.  Thus spans can be 
    taken to be of the form $(u,1)$ with $Y \to^{u} X$.  
  
    The isomorphism  then follows.
  \end{proof}

Note that $\X$ and $\X^{\ast\ast}$ are also isomorphic as fibrations over $\mathbb{B}$; see \cite[Proposition 3.4]{arxiv:simple-dual}. 

\subsection{Dagger fibrations}


We now introduce the notion of a dagger fibration. First recall that a \emph{morphism of fibrations} (over a fixed base) is a commuting triangle: 
$$  \begin{tikzcd}
    \X \ar[dr,"p"'] \ar[rr,"h"] && \Y \ar[dl,"q"] \\
    & \mathbb{B}
  \end{tikzcd}  $$
where $h$ carries Cartesian maps to Cartesian maps.  

\begin{definition}
A {\bf dagger fibration} is 
given by a fibration $\X \to^{q} \mathbb{B}$ with a morphism of fibrations $\X \to^{(\blank)^\dagger} \X^\ast$ such that
$$  \begin{tikzcd}
    \X \ar[r,"(\blank)^\dagger"] \ar[dr,"q"'] \ar[rr,bend left,shift left,"1_\X"]
     & \X^\ast \ar[d,"q^\ast"] \ar[r,"(\blank)^\dagger"] & \X^{\ast\ast}=\X \ar[dl,"q"]\\
    & \mathbb{B}
  \end{tikzcd}  $$
and such that $\dagger$ is stationary on objects.  A dagger fibration has a {\bf dagger cleavage} 
when $(\blank)^\dagger$ sends cloven cartesian arrows to cloven cartesian arrows.
\end{definition}

Our main example of a dagger fibration will be the linear fibration of a Cartesian reverse differential category.  We begin by defining the required dagger (this is a more general form of the dagger discussed earlier in Section \ref{section:dagger-cats}):

\begin{definition}\label{defn:ContextDagger}
In a Cartesian reverse differential category $\X$, for a map $C\x A \to^{f} B$, define the {\bf contextual $\dagger$} of $f$, $C \times B \to^{f^{\dagger[C]}} A$, as follows:
  $$  f^{\dagger[C]} := C\x B \to^{\iota_0 \x 1} (C\x A)\x B \to^{\R[f]} C\x A \to^{\pi_1} A  $$
\end{definition}

\begin{lemma}\label{lemma:dagger-context-linear}
  In a Cartesian reverse differential category, for any map $C \x A \to^{f} B$, the following
  are equivalent:
  \begin{enumerate}
    \item $f$ is linear in $A$ (Definition \ref{definition:linear}) with respect to the differential combinator of Theorem \ref{theorem:cdc-from-rdc};
    \item $(\iota_0 \x \iota_1)\ex \D[f] = f$;
    \item $f^{\dagger[C]\dagger[C]} = f$.
  \end{enumerate}
\end{lemma}
\begin{proof}
  \emph{1} $\Leftrightarrow$ \emph{2} follows from Lemma 
  \ref{lemma:partial-linear-equiv}.  To show that 
  \emph{2} $\Leftrightarrow$ \emph{3} requires a bit more work, but the proof is essentially the same as in Lemma \ref{lemma:linear-reverse}. \end{proof}


\begin{corollary}\label{cor:dagger-preserves-linear-arg}
  Let $\X$ be a Cartesian reverse differential category and let $I \x A \to^{g} B$ be linear in $A$.  Then $I\x B \to^{g^{\dagger[I]}} A$
  is linear in $B$.
\end{corollary}

\begin{theorem}\label{thm:ReverseHasDualFibration}
If $\X$ is a Cartesian reverse differential category, then its associated linear fibration is a dagger fibration, with dagger as in Definition \ref{defn:ContextDagger}.  
\end{theorem}
\begin{proof}

    First, we must show that the assignment
    \[
      \widetilde{\mathsf{Lin}(\X)} \to^{(\blank)^{\dagger}} \widetilde{\mathsf{Lin}(\X)}^\ast
    \]
    given by
    \[
      ((I,A) \to^{(f,g)} (J,B))^\dagger := (f,g^{\dagger[I]})  
    \]
    is a morphism of fibrations where $g^{\dagger[I]}$ is the contextual 
    $\dagger$ of Definition \ref{defn:ContextDagger}.   This assignment 
    is well-defined by Corollary \ref{cor:dagger-preserves-linear-arg}, and is by definition stationary on objects.  
    
    First, we show it is a functor.  That it preserves identities: we have
    $(1,\pi)^\dagger = (1,\pi_1^{\dagger[I]})$.  Thus it suffices to show 
    that $\pi_1^{\dagger[I]} = \pi_1$, but 
    \[
     \pi_1^{\dagger[I]}
       = (\<1,0\> \x 1)R[\pi_1]\pi_1
       = (\<1,0\> \x 1)\pi_1 \<0,1\>\pi_1
       = \pi_1 
    \]
    as desired.
  
    Next, we show that it preserves composition.  We begin with:
    \begin{align*}
      ((f,g)(f',g'))^\dagger 
      &= (ff',(\<\pi_0f,g\>g')^{\dagger[I]}).
    \end{align*}
    Next, 
    \[
      (\<\pi_0f,g\>g')^{\dagger[I]} = (\<1,0\> \x 1)\R[\<\pi_0f,g\>g'] \pi_1.
    \]
    We first isolate the middle piece:
    \begin{align*}
      &  \R[\<\pi_0f,g\>g']\\
      &= \<\pi_0,\<\pi_0\<\pi_0f,g\>,\pi_1\>\R[g']\>\R[\<\pi_0f,g\>]\\
      &= \<\pi_0,\<\pi_0\<\pi_0f,g\>,\pi_1\>\R[g']\>((1\x \pi_0)\R[\pi_0f]+(1\x \pi_1)\R[g])\\
      &= \<\pi_0,\<\pi_0\<\pi_0f,g\>,\pi_1\>\R[g']\>((\pi_0 \x \pi_0)\R[f]\iota_0 + (1\x \pi_1)\R[g])
    \end{align*}
    Now when we postcompose the above by $\pi_1$ the first piece of the sum vanishes, because $\iota_0\pi_1=0$.
    Thus, we resume the main calculation of $(\<\pi_0f,g\>g')^{\dagger[I]}$:
    \begin{align*}
      & (\<1,0\> \x 1)\R[\<\pi_0f,g\>g'] \pi_1\\
      &= (\<1,0\> \x 1) \<\pi_0,\<\pi_0\<\pi_0f,g\>,\pi_1\>\R[g']\>((\pi_0 \x \pi_0)\R[f]\iota_0 + (1\x \pi_1)\R[g])  \pi_1\\
      &=(\<1,0\> \x 1) \<\pi_0,\<\pi_0\<\pi_0f,g\>,\pi_1\>\R[g']\>(1\x \pi_1)\R[g] \pi_1\\
      &=(\<1,0\> \x 1) \<\pi_0,\<\pi_0\<\pi_0f,g\>,\pi_1\>\R[g']\pi_1\>\R[g] \pi_1\\
      &= \<\pi_0\<1,0\>,\<\pi_0\<1,0\>\<\pi_0f,g\>,\pi_1\>\R[g']\pi_1\>\R[g]\pi_1\\
      &= \<\pi_0\<1,0\>,\<\pi_0\<f,\<1,0\>g\>,\pi_1\>\R[g']\pi_1\>\R[g]\pi_1\\
      &= \<\pi_0\<1,0\>,\<\pi_0\<f,0\>,\pi_1\>\R[g']\pi_1\>\R[g]\pi_1 \qquad \text{$g$ is linear in $2^\text{nd}$ arg}\\
      &= \<\pi_0\<1,0\>,\<\pi_0f\<1,0\>,\pi_1\>\R[g']\pi_1\>\R[g]\pi_1 \\
      &= \<\pi_0\<1,0\>,\<\pi_0f,\pi_1\>(\<1,0\> \x 1)\R[g']\pi_1\>\R[g]\pi_1 \\
      &= \<\pi_0,\<\pi_0f,\pi_1\>\>(\<1,0\> \x ((\<1,0\>\x 1)\R[g']\pi_1))\R[g]\pi_1\\
      &= \<\pi_0,\<\pi_0f,\pi_1\>\>(1\x ((\<1,0\> \x 1)\R[g']\pi_1) )\R[g]\pi_1
    \end{align*}
  
    Now consider 
    \begin{align*}
      & (f,g)^\dagger (f',g')^\dagger \\
      &= (f,(\<1,0\> \x 1)\R[g]\pi_1) (f',(\<1,0\>\ x1)\R[g']\pi_1)\\
      &= (ff',\<\pi_0,\<\pi_0f,\pi_1\>\>(1\x ((\<1,0\> \x 1)\R[g']\pi_1) )\R[g]\pi_1)\\
      &= (ff',(\<1,0\> \x 1)\R[\<\pi_0f,g\>g'] \pi_1) \qquad \text{by the above}\\
      &= ((f,g)(f',g'))^\dagger
    \end{align*}
  
    Thus $(\blank)^\dagger$ preserves composition, hence is a functor.
  
    Next, 
    \[
      \begin{tikzcd}
         \widetilde{\mathsf{Lin}(\X)} \ar[dr,"\pi"']\ar[rr,"{(\blank)^\dagger}"] && \widetilde{\mathsf{Lin}(\X)}^\ast \ar[dl,"\pi^\ast"] \\ 
           & \X 
      \end{tikzcd}
    \]
    commutes because $\pi^\ast((f,g)^\dagger) = \pi^\ast(f,g^{\dagger[I]}) = f = \pi(f,g)$.
  
    We have already seen that $\pi_1^{\dagger[I]}= \pi_1$, thus $(\blank)^\dagger$ 
    carries Cartesian morphisms to Cartesian morphisms, thus it is a morphism of 
    fibrations.  Also note that the above fact means that $\dagger$ is stationary on 
    Cartesian arrows: $(f,\pi_1)^\dagger = (f,\pi_1)$, and hence stationary on objects, and 
    the fibration has a dagger cleavage.
  
    Finally, note that for a map $(f,g) : (I,A) \to (J,B)$ we require that 
    $g$ be linear in $A$.  Then because $g$ is linear in $A$ 
    \[
      (f,g)^{\dagger\dagger} = (f,g^{\dagger[I]})^\dagger = (f,g^{\dagger[I]\dagger[I]}) = (f,g)
    \]
    by Lemma \ref{lemma:dagger-context-linear}.  Thus, the linear fibration of $\X$ is a dagger fibration.  
  \end{proof}

\begin{lemma}\label{lemma:fibration-dagger-lemma}
  If $\X \to^{q} \mathbb{B}$ is a dagger fibration with a dagger cleavage, then each fiber $q^{\text{-}1}(A)$ 
  is a $\dagger$-category, and reindexing preserves the dagger.
\end{lemma}
\begin{proof}
    First, ${q^*}^{-1}(A)$ is the category whose objects are those of $q^{-1}(A)$, and whose morphisms $X \to Y$ are 
    spans of the form $X \from^{h} Y = Y$.  These then correspond isomorphically to maps $Y \to^{h} X$ in $q^{-1}(A)$,
    and in fact there is an isomorphism of categories $\alpha_A$ that sends $(v,1) \mapsto v$.  The dagger on 
    $q^{-1}(A)$ is defined by the following diagram:
    \[
        \begin{tikzcd}[row sep=3em]
            q^{-1}(A) \ar[r,"(\blank)^\dagger_A"] \ar[dr,"(\blank)^{\dagger[A]}"'] & {q^*}^{-1}(A) \ar[d,"\alpha_A"] \\
            {} \ar[ur,"{:=}",phantom,near end]& q^{-1}(A)^\text{op}
        \end{tikzcd}
    \]
    The isomorphism $\alpha_A$ also induces a reindexing for opposite fibers:
    \[
        u^* \, := \, q^{-1}(A)^\text{op} \to^{\alpha_A^{-1}} {q^*}^{-1}(A) \to^{u^*} {q^*}^{-1}(B) \to^{\alpha_B} q^{-1}(B)^\text{op}
        \]

    Then consider the following diagram:
    \[
        \begin{tikzcd}
            q^{-1}(A) \ar[d,"u^*"'] \ar[rr,bend left,"(\blank)^{\dagger[A]}"] \ar[r,"(\blank)^\dagger_A"] & {q^*}^{-1}(A) \ar[d,"u^*"]\ar[r,"\alpha_A"] & q^{-1}(A)^\text{op} \ar[d,"u^*"]\\
            q^{-1}(B) \ar[rr,bend right,"(\blank)^{\dagger[B]}"'] \ar[r,"(\blank)^\dagger_B"'] & {q^*}^{-1}(B) \ar[r,"\alpha_A"'] & q^{-1}(B)^\text{op}
        \end{tikzcd}
        \]
    The right square and top and bottom triangles commute definitionally.  The commutativity of the left 
    square follows from the fact that $(\blank)^\dagger$ sends cloven cartesians to cloven cartesians.
  \end{proof}

\subsection{Characterization of Cartesian reverse differential categories}
We have seen in the previous sections that a Cartesian reverse differential category is a Cartesian differential category whose associated linear fibration is a dagger fibration in which each fibre has $\dagger$-biproducts.  In this final section, we show that this collection of structures characterizes Cartesian reverse differential categories.  

\begin{definition}\label{definition:context-dagger}
  Let $\X$ be a Cartesian differential category.  We say that $\X$ has a 
  {\bf contextual linear dagger} when the linear fibration is a dagger fibration
$$    \begin{tikzcd} 
      \widetilde{\mathsf{Lin}(\X)} \ar[dr,"\pi"'] \ar[rr,"{(\blank)^\dagger}"] && \widetilde{\mathsf{Lin}(\X)}^\ast \ar[dl,"\pi^\ast"]\\
       & \X 
    \end{tikzcd} $$
  and each fiber category $\mathsf{Lin}(\X)[A]$ has $\dagger$-biproducts.
 \end{definition}

By Lemma \ref{lemma:fibration-dagger-lemma}, every fiber of such a fibration is a 
$\dagger$-category, and reindexing functors preserve the dagger.
We denote the $\dagger$ in the fiber 
$\mathsf{Lin}(\X)[A]$ by $(\blank)^{\dagger[A]}$. In particular we note that $(\blank)^{\dagger[A]}$ preserves the additive structure. Before giving the main theorems of this section, we will need the following lemma: 

\begin{lemma}\label{lemma:cld.2}
  Let $\X$ be a Cartesian differential category with a 
  contextual linear dagger. 
  For any map $A \to^{f} B$ the following diagram commutes.
       $$      \begin{tikzcd}[column sep=8em]
              (A\x B) \x A \ar[r,"{\D[\D[f]^{\dagger[A]}]^{\dagger[A\x B]}}"] & A\x B \ar[d,"\pi_1"]\\
              A\x A \ar[u,"{\<1,0\> \x 1}"] \ar[r,"{\D[f]}"'] & B
             \end{tikzcd}  $$ 
  \end{lemma}
  As done in the proof of Lemma \ref{lemma:partial-linear-equiv}, we will 
  distinguish maps $f : A \x B \to C$ as maps in $\mathsf{Lin}(\X)[A]$ by 
  underlining them $\underline{f} : B\to C$. 
   \begin{proof}
      First note that for any $C \x A \to^{f} B$ in $\mathsf{Lin}(\X)[C]$, that is, $f$ is linear in $A$, we have the following equalities: 
          \begin{align*}
        f^{\dagger[C]} &= ((\iota_0 \times \iota_1)\D[f])^{\dagger{C}}  \tag{Lemma \ref{lemma:partial-linear-equiv}}\\
                       &= (\iota_0 \x 1) \left((1\x \iota_1)\D[f] \right)^{\dagger[C]} \tag{reindexing preserves $\dagger$}\\
                       &= (\iota_0 \x 1) \left(\underline{\iota_1 \D[f]} \right)^{\dagger[A]} \\
                       &= (\iota_0 \x 1) \underline{\D[f]^{\dagger[C\x A]}\pi_1} \tag{$\dagger$-biproducts in the fiber}\\
                       &= (\iota_0 \x 1) \<\pi_0,\D[f]^{\dagger[C\x A]}\>\pi_1\pi_1\\
                       &= (\iota_0 \x 1)\D[f]^{\dagger[C\x A]}\pi_1 
      \end{align*}
  So $f^{\dagger[C]} = (\iota_0 \x 1)\D[f]^{\dagger[C\x A]}\pi_1$. 
      Now, let $A \to^{f} B$ be any map.  Note that $\D[f]$ is linear in its second $A$, 
      and thus $\D[f]^{\dagger[A]\dagger[A]} = \D[f]$.  Then applying the above result
      to $\D[f]^{\dagger[A]\dagger[A]}$ we get $\D[f] = (\iota_0 \x 1)\D[\D[f]^{\dagger[A]}]^{\dagger[A\x B]}$
      as required.
    \end{proof}

\begin{theorem}\label{theorem:rdc-from-cdc}
  A Cartesian differential category $\X$ with a contextual linear dagger is a 
  Cartesian reverse differential category with reverse differential combinator $\R$ defined as follows (for a map $A \to^f B$): 
$$ \R[f] := A\x B \to^{\D[f]^{\dagger[A]}} B $$
\end{theorem}
\begin{proof}
    We define the reverse differential combinator as follows
    \[
      \infer{A\x B \to_{\R[f] := \D[f]^{\dagger[A]}} A}{A\to^{f} B}  
    \]
    noting the above makes sense because $\D[f]$ is linear in the second $A$.
  
    \begin{enumerate}[{\bf [RD.1]}]
      \item The calculation is as follows 
            \begin{align*}
              \R[f+g] &= \D[f+g]^{\dagger[A]} \\
                     &= (\D[f]+ \D[g])^{\dagger[A]} \\
                     &= \D[f]^{\dagger[A]} + \D[g]^{\dagger[A]} \\
                     &= \R[f]+\R[g]
            \end{align*}
            Similarly, $\R[0]=0$.
      \item Note that linear implies additive and the typing of $\dagger$ on a 
            a fiber sends maps that are linear in their second argument to maps that 
            are linear in their second argument.  In particular $\D[f]^{\dagger[A]}$ Is
            linear in its second argument.  Thus 
            \[
              \<a,b+c\>\R[f] = \<a,b+c\>\D[f]^{\dagger[A]}= \<a,b\>\R[f] + \<a,c\>\R[f]  
            \]
            Similarly, $\<a,0\>\R[g]= 0$.
      \item To show that $\R[1] = \pi_1$, first note that $\dagger[A]$ is a functor 
            in the fiber over $A$.  In particular $\underline{1}^{\dagger[A]} = \underline{1}$ but 
            $\underline{1} = \pi_1$.  Then we note that 
            \[
              \R[1] = (\D[1])^{\dagger[A]} = \pi_1^{\dagger[A]} = \underline{1}^{\dagger[A]} = \underline{1}=\pi_1  
            \]
            Similarly, $\dagger$ is a gives a $\dagger$-biproduct structure in each 
            fiber, hence $\underline{\pi_i}^{\dagger[A\x B]} = \underline{\iota_i}$.
            Then,
            \[
            \R[\pi_0] = (\D[\pi_0])^{\dagger[A\x B]}
                     = (\pi_1\pi_0)^{\dagger[A\x B]}
                     = \underline{\pi_0}^{\dagger[A\x B]}
                     = \underline{\iota_0}
                     = \pi_1 \iota_0  
            \]
            as desired.  Similarly, $\R[\pi_1] = \pi_1\iota_1$.
      \item We have the following calculation, where we note that the pairing in  
            a fiber is the pairing of the maps in the underlying category.
            \begin{align*}
              \R[\<f,g\>]
                &= \D[\<f,g\>]^{\dagger[A]} \\
                &= (\<\D[f],\D[g]\>)^{\dagger[A]}\\
                &= (\underline{\<\D[f],\D[g]\>})^{\dagger[A]}\\
                &= (\underline{\D[f]\iota_0 + \D[g]\iota_1})^{\dagger[A]}\\
                &= \underline{\pi_0 \D[f]^{\dagger[A]} + \pi_1 \D[g]^{\dagger[A]}} \qquad \text{$\dagger$ is contravariant}\\
                &= \<\pi_0,\pi_1\pi_0\>\D[f]^{\dagger[A]} + \<\pi_0,\pi_1\pi_1\>\D[g]^{\dagger[A]}\\
                &= (1\x \pi_0) \R[f] + (1\x \pi_1)\R[g]
            \end{align*}
      \item Here we use that {\bf [RD.5]} is equivalent to asking that the assignment 
            \[
              \X \to^{\R} \widetilde{\mathsf{Lin}(\X)}^\text{op} ; f \mapsto (f,R[f])  
            \]      
            be functorial.  Also, {\bf [CDC.5]} says that 
            \[
              \X \to^{D}   \widetilde{\mathsf{Lin}(\X)} ; f \mapsto (f,\D[f])
            \]
            is functorial.  Then we have 
            \[
              \X \to^{D} \widetilde{\mathsf{Lin}(\X)} \to^{\dagger} \widetilde{\mathsf{Lin}(\X)}^\text{op}
            \]
            is the assignment $f \mapsto (f,\R[f])$.  Hence, as functors compose, {\bf [RD.5]} holds.
      \item Our goal is to show that 
        \[\<1\x \pi_0,0\x \pi_1\>(\<1,0\> \x 1)R^{(3)}[f]\pi_1 = (1\x \pi_1)R[f]\]
        Here we use the coherence {Lemma \ref{lemma:cld.2}}:
            \begin{align*}
              & \<1\x\pi_0,0\x \pi_1\>(\<1,0\>\x 1)\R^{(3)}[f]\pi_1 \\
              &= \<1\x\pi_0,0\x \pi_1\>(\<1,0\>\x 1) \D[ \D[ \D[f]^{\dagger[A]} ]^{\dagger[A\x B]} ]^{\dagger[(A\x B)\x A]}\pi_1\\
              &= \<1\x\pi_0,0\x \pi_1\> \D[\D[f]^{\dagger[A]}]
            \end{align*}
        Now we invoke the fact that dagger sends maps that are linear in their second argument
        to maps that are linear in their second argument.  Thus $\D[f]^{\dagger[A]}$ is linear in 
        its second argument.  But then that means 
        \[
         \<1\x \pi_0,0\x \pi_1\>\D[\D[f]^{\dagger[A]}] = (1\x \pi_1) \D[f]^{\dagger[A]}
           = (1\x \pi_1)\R[f]
        \]
        as desired.
      \item This is immediate from Lemma \ref{lemma:cld.2}: applying it 
            twice to both sides gives $\D[\D[f]] = \ex \D[\D[f]]$ which holds by {\bf [CDC.7]}.
    \end{enumerate}
  \end{proof}


We conclude with the main result of this paper: 

\begin{theorem}\label{thm:reverseDerivEquiv}
A Cartesian reverse differential category is precisely a Cartesian 
  differential category with a contextual linear dagger.
\end{theorem}
\begin{proof}
    Let $\X$ be a Cartesian reverse differential category.  Then $\X$ is a Cartesian 
    differential category by Theorem \ref{theorem:cdc-from-rdc}, its associated linear fibration is a dagger fibration by Theorem \ref{thm:ReverseHasDualFibration}, and each fibre has $\dagger$-biproducts by Proposition \ref{prop:reverseGivesDaggerBiproducts}.  
  
    Conversely, if $\X$ is a Cartesian differential category with contextual linear 
    dagger, then $\X$ is a reverse differential category by Theorem \ref{theorem:rdc-from-cdc}.

    The only thing left to show is that the constructions of reverse derivatives and 
    Cartesian derivatives used in the above are inverse to each other.
  
    First, on the one hand, if we start with a Cartesian differential category with 
    contextual linear dagger, form the reverse derivative from this, then
    form a Cartesian derivative from the induced reverse derivative, Lemma \ref{lemma:cld.2} says that the resulting induced Cartesian derivative
    structure is the starting differential structure.
  
    On the other hand, suppose that we start with a reverse derivative, and define the 
    Cartesian derivative by $\D[f] = (\<1,0\> \x 1)\R^{(2)}[f]\pi_1$.  Then,
    after this we use the induced contextual $\dagger$ of Definition \ref{definition:context-dagger},
    to define a reverse derivative.  This has 
    \[
      \infer{
        \infer{
          \D[f]^{\dagger[A]}:=
          (\<1,0\> \x 1)\R[(\<1,0\> \x1)\R^{(2)}[f]\pi_1]\pi_1
        }{
          A\x A \to^{\<1,0\> \x 1}(A\x B)\x A \to^{\R^{(2)}[f]} A\x B \to^{\pi_1 }B
        }
      }{A\to^{f} B}  
    \]
    and we want to show that $\D[f]^{\dagger[A]} = R[f]$.  Yet,
    \begin{align*}
      & (\<1,0\> \x 1)\R[(\<1,0\> \x1)\R^{(2)}[f]\pi_1]\pi_1 \\
      &= (\<1,0\> \x 1)((\<1,0\> \x 1)\x 1)\R[\R^{(2)}[f]\pi_1](\pi_0 \x 1)\pi_1\\
      &= (\<1,0\> \x 1)((\<1,0\> \x 1)\x 1)(1\x \iota_1)\R^{(3)}[f]\pi_1\\
      &= (\<\<1,0\>,0\>\x \<0,1\>)\R^{(3)}[f]\pi_1 \\
      &= \R[f] \qquad \text{Lemma \ref{lemma:rd6-alternative-form}}
    \end{align*}
  \end{proof}

\section{Concluding remarks}\label{section:conclusion}
This paper begins the story of categories with a reverse derivative; however,
 there is much more that needs to be done in this area.  Perhaps the most
  important next step is to add partiality into this setting.  One way to add partiality to categories is via a restriction structure
   \cite{journal:rcats1}.  The paper \cite{journal:diff-rest} showed how
    to combine a Cartesian differential structure with a restriction structure 
    to obtain ``differential restriction categories.''  This provides an axiomatization for categories of
     smooth partial maps.  A key next step is then to combine reverse 
     differential categories with restriction structure, and check that many 
     of the results that held for differential restriction categories hold 
     for ``reverse differential restriction categories''.  Such a structure 
     would bring us even closer to a true categorical semantics for differential 
     programming.

Another important aspect to develop will be the term logic for reverse 
differential categories.  The term logic for Cartesian differential categories 
greatly facilitates the ability to establish and prove results in that abstract setting; a 
term logic for  reverse differential categories is similarly important.

Tensors are another important aspect of differential programming, and form the 
foundations on which modern, large scale machine learning platforms are based 
\cite{tensorflow2015-whitepaper}.  In \cite{journal:BCS-Storage}, monoidal structure was described in a way that
interacts well with differentiation.  In particular, $V\ox W$ is the object
for which bilinear maps $V\x W \to U$ correspond to linear maps $V\ox W \to U$.  Developing a similar structure for the reverse derivative will thus also be important.  More generally, there should be a notion of (monoidal) reverse differential category.  These should provide additional examples of Cartesian reverse differential categories: just as the coKleisli category of a (monoidal) differential category \cite{DiffCats} is a Cartesian differential category, so should the coKleisli category of a monoidal reverse differential category be a Cartesian reverse differential category.

Finally, an important generalization of Cartesian differential categories are tangent
 categories \cite{journal:TangentCats}, a categorical setting for differential
  geometry which axiomatizes the existence of a ``tangent bundle'' for each object.
      Every Cartesian differential category gives rise to a tangent category.  
      A reverse derivative category should give a 
      ``category with a cotangent bundle for each object''; defining such categories will be another important extension of this work.

\bibliography{ReverseDerivatives-Arxiv}

\begin{thebibliography}{10}

\bibitem{tensorflow2015-whitepaper}
Mart\'{\i}n Abadi, Ashish Agarwal, Paul Barham, Eugene Brevdo, Zhifeng Chen,
  Craig Citro, Greg~S. Corrado, Andy Davis, Jeffrey Dean, Matthieu Devin,
  Sanjay Ghemawat, Ian Goodfellow, Andrew Harp, Geoffrey Irving, Michael Isard,
  Yangqing Jia, Rafal Jozefowicz, Lukasz Kaiser, Manjunath Kudlur, Josh
  Levenberg, Dandelion Man\'{e}, Rajat Monga, Sherry Moore, Derek Murray, Chris
  Olah, Mike Schuster, Jonathon Shlens, Benoit Steiner, Ilya Sutskever, Kunal
  Talwar, Paul Tucker, Vincent Vanhoucke, Vijay Vasudevan, Fernanda Vi\'{e}gas,
  Oriol Vinyals, Pete Warden, Martin Wattenberg, Martin Wicke, Yuan Yu, and
  Xiaoqiang Zheng.
\newblock {TensorFlow}: Large-scale machine learning on heterogeneous systems,
  2015.
\newblock Software available from tensorflow.org.
\newblock URL: \url{https://www.tensorflow.org/}.

\bibitem{journal:BCS:CDC}
R.~Blute, R.~Cockett, and R.~Seely.
\newblock {Cartesian Differential Categories}.
\newblock {\em Theory and Applications of Categories}, 22:622--672, 2009.

\bibitem{journal:BCS-Storage}
R.~Blute, R.~Cockett, and R.~Seely.
\newblock {Cartesian Differential Storage Categories}.
\newblock {\em Theory and Applications of Categories}, 30(18):620--686, 2015.

\bibitem{DiffCats}
R.F. Blute, J.R.B. Cockett, and R.A.G. Seely.
\newblock Differential categories.
\newblock {\em Mathematical structures in computer science}, 16(6):1049--1083,
  2006.

\bibitem{book:categoricalalgebra}
F.~Borceaux.
\newblock {\em {Handbook of categorical algebra II}}.
\newblock Cambridge University Press, 2008.

\bibitem{chr12}
Bruce Christianson.
\newblock A {L}eibniz notation for automatic differentiation.
\newblock In {\em Recent Advances in Algorithmic Differentiation}, volume~87 of
  {\em Lecture Notes in Computational Science and Engineering}, pages 1--9.
  Springer, 2012.

\bibitem{journal:diff-rest}
J.R.B. Cockett, G.S.H. Cruttwell, and J.D. Gallagher.
\newblock Differential restriction categories.
\newblock {\em Theory and applications of categories}, 25(21):537--613, 2011.

\bibitem{journal:rcats1}
J.R.B. Cockett and Stephen Lack.
\newblock Restriction categories i: categories of partial maps.
\newblock {\em Theoretical Computer Science}, 270(1):223 -- 259, 2002.
\newblock URL:
  \url{http://www.sciencedirect.com/science/article/pii/S0304397500003820},
  \href {http://dx.doi.org/https://doi.org/10.1016/S0304-3975(00)00382-0}
  {\path{doi:https://doi.org/10.1016/S0304-3975(00)00382-0}}.

\bibitem{journal:TangentCats}
R.~Cockett and G.~Cruttwell.
\newblock {Differential structure, tangent structure, and SDG}.
\newblock {\em Applied Categorical Structures}, 22:331--417, 2014.

\bibitem{elliott-AD-ICFP}
Conal Elliott.
\newblock The simple essence of automatic differentiation.
\newblock {\em Proceedings of the ACM on Programming Languages}, 2(ICFP):70,
  2018.

\bibitem{griewank2012invented}
Andreas Griewank.
\newblock Who invented the reverse mode of differentiation.
\newblock {\em Documenta Mathematica, Extra Volume ISMP}, pages 389--400, 2012.

\bibitem{book:Jacobs-Cat-Log}
B.~Jacobs.
\newblock {\em {Categorical logic and type theory}}.
\newblock Number 141 in {Studies in logic and the foundations of mathematics}.
  Elsevier, 1999.

\bibitem{arxiv:simple-dual}
Anders {Kock}.
\newblock {The dual fibration in elementary terms}.
\newblock {\em arXiv e-prints}, page arXiv:1501.01947, Jan 2015.
\newblock \href {http://arxiv.org/abs/1501.01947} {\path{arXiv:1501.01947}}.

\bibitem{lemay2018tangent}
J-S~P. Lemay.
\newblock A tangent category alternative to the faa di bruno construction.
\newblock {\em Theory and Applications of Categories}, 33(35):1072--1110, 2018.

\bibitem{linnainmaa}
S.~Linnainmaa.
\newblock {Taylor expansion of the accumulated rounding error.}
\newblock {\em BIT Numerical Mathematics}, 16:146--160, 1976.

\bibitem{talk:plotkin-2018}
G.~Plotkin.
\newblock A simple differential programming language.
\newblock MFPS 2018 Keynote Address, June 2018.

\bibitem{backpropagation}
David~E Rumelhart, Geoffrey~E Hinton, and Ronald~J Williams.
\newblock Learning representations by backpropagating errors.
\newblock {\em Cognitive modeling}, 5:3, 1988.
\newblock URL: \url{www.cs.toronto.edu/hinton/naturebp.pdf}.

\bibitem{journal:selinger-dagger}
Peter Selinger.
\newblock Dagger compact closed categories and completely positive maps.
\newblock {\em Electron. Notes Theor. Comput. Sci.}, 170:139--163, March 2007.
\newblock URL: \url{http://dx.doi.org/10.1016/j.entcs.2006.12.018}, \href
  {http://dx.doi.org/10.1016/j.entcs.2006.12.018}
  {\path{doi:10.1016/j.entcs.2006.12.018}}.

\end{thebibliography}

\end{document}